\documentclass[submission,copyright,creativecommons]{eptcs}
\providecommand{\event}{ICE 2012}
\usepackage{syntax}
\usepackage{floatflt}
\usepackage{wrapfig}
\usepackage{framed}
\usepackage{natbib}

\usepackage{amsmath}
\usepackage{amssymb}
\usepackage{stmaryrd}
\usepackage{graphicx}
\usepackage{bussproofs}
\usepackage{url}
\usepackage{amsthm}
\usepackage{mdwlist}

\newcommand\GoS{\textsf{GoS}}
\newcommand\Verity{\textsf{Verity}}
\newcommand\Algol{\textsf{Algol}}

\newcommand{\sbr}[1]{\llbracket #1 \rrbracket}

\newcommand{\noincludegraphics}[1]{}

\newtheorem{definition}{Definition}
\newtheorem{theorem}{Theorem}
\newtheorem{lemma}{Lemma}

\begin{document}

%15 pages plus references

\title{Coherent Minimisation:\\ Towards efficient tamper-proof compilation}

\author{Dan R. Ghica
       \institute{University of Birmingham}
       \email{d.r.ghica@cs.bham.ac.uk}
       \and
			  Zaid Al-Zobaidi
       \institute{University of Birmingham}
       \email{z.k.ibrahim@cs.bham.ac.uk}
       }
%\author{Dan R. Ghica \and Zaid Al-Zobaidi}

\def\titlerunning{Coherent Minimisation}
\def\authorrunning{D.R.\ Ghica \& Z.\ Al-Zobaidi}
\maketitle

\begin{abstract}
Automata representing game-semantic models of programs are meant to operate in environments whose input-output behaviour is constrained by the \emph{rules of a game}. This can lead to a notion of equivalence between states which is weaker than the conventional notion of bisimulation, since not all actions are available to the environment. An environment which attempts to break the rules of the game is, effectively, mounting a \emph{low-level attack} against a system. In this paper we show how (and why) to enforce game rules in games-based hardware synthesis and how to use this weaker notion of equivalence, called \emph{coherent equivalence}, to aggressively minimise automata.
\end{abstract}

\section{Introduction}  

Computer security `exploits' take advantages of mistakes in programs, called `vulnerabilities', to cause unintended behaviour to occur on a computing device. Exploits are most commonly low-level attacks that violate the abstractions of the programming language to create behaviour inexpressible in the language itself. Such attacks are possible because lower-level languages (`machine code') are less constrained behaviourally than higher-level languages, so a run-time system, when confronted with executable code, cannot tell whether that code is the result of a legitimately compiled program or whether it contains behaviours deemed `illegal'. Restricting the behaviour of machine code is the essence of `tamper-resistant' compilation, and it can be achieved in various way: sand-boxing the code to prevent unauthorised access to memory, randomising the memory layout so that code cannot `guess' where certain data is stored even if it has physical access to it~\citep{ShachamPPGMB04} or monitoring the control flow in a program to ensure that no arbitrary jumping occurs~\citep{AbadiBEL09}. 

A compiler and runtime system that can detect and enforce machine code behaviour so that it satisfies all the abstraction of the higher-level programming language would be, effectively, a `fully abstract' compilation and execution environment offering the maximum level of tamper resistance: `tamper-proof' compilation~\citep{Abadi98}. In a general-purpose system this is perhaps impossible to achieve in a practical way. However, we will show how it is achievable in `higher-level synthesis' (also known as `hardware compilation'), the automatic synthesis of special-purpose digital circuits from programs written in conventional programming languages, using game-semantic models. We will then further show how, once guarantees on the behaviour of the environment can be effectively enforced, the automata representation of the game semantic models can be aggressively optimised. Essentially, an environment which cannot perform arbitrary actions cannot distinguish between as many states in a transition system as the unrestricted system. This is a new notion of equivalence, which we call `coherent equivalence'.

\section{The GoS hardware compiler} %2p

The \emph{Geometry of Synthesis} (\GoS)\footnote{\url{http://veritygos.org}} compiler~\citep{Ghica07,GhicaS10,GhicaS11,GhicaSS11} produces (VHDL) descriptions of digital circuits from a conventional functional-imperative programming language. The circuits produced by the compiler are a concrete representation of the game-semantic model of the language~\citep{GhicaMO06}.

\subsection{The language Verity}

The source language of \GoS\ is called \Verity, and it is an \Algol-like language in \cite{Reynolds81} sense. It represents a combination of the simply-typed (call-by-name) lambda calculus with the simple imperative language of while loops. Additionally, \Verity\ has primitives for parallel execution of commands. 

The combination of call-by-name and local store, although made popular in \Algol\ 60, fell out of favour as languages with global store (and more generally, global effects) and call-by-reference (\textsf{C}), call-by-value (\textsf{ML}) and call-by-need (\textsf{Haskell}) became prevalent for reasons of convenience and efficiency. 

However, in the case of hardware compilation the perceived disadvantages of \Algol\ yield unexpected benefits:
\begin{description}
\item[Local store.] The notion of global store does not fit the way memory is used in a circuit. In a circuit, stateful elements are scattered throughout the design, wherever needed. There is no need to bring them all together in a single global memory because this would be inefficient in multiple ways. Managing access to this global memory would require complex control elements which would be costly in energy, footprint and latency. It would also constitute a bottleneck for concurrency. Note that the lack of language support for global store does not mean that \Verity\ cannot deal with programs which access off-chip RAM. It only means that such access needs to be programmed explicitly and used via library calls. This is an advantage because RAM controllers can exploit the precise memory hierarchy of the device in a way that generic language support cannot. 
\item[Call by name.] \Verity\ is a functional programming language, and it is well known that managing closures is one of the great potential sources of inefficiency in compilers. Dealing with memory management for closures in functional hardware synthesis raises additional difficulties because all usage of memory in a circuit must be bounded at synthesis time. This makes it impossible to support higher order functions~\citep{MycroftS00}. However, call-by-name closures require less storage, because of constant re-evaluation of the thunks. This provides an elegant, albeit somewhat fortuitous, solution to the problem of memory management for closures. 
\end{description}
The syntax of the language is standard for an \Algol-like language. Here we only provide two examples, to give a flavour of the language. First, a naive and highly inefficient implementation of a Fibonacci number calculator:
\begin{verbatim}
let fbn = (fix \f.\x. if x<1 then 0 else if x<2 then 1 else f(x-1)+f(x-2)) in fbn(5)
\end{verbatim}
Second, an efficient implementation using memoisation:
\begin{verbatim}
  new mem(128) in
  new i := 0 in
  while !i < 128 do {mem(!i) := 0; i := !i + 1};
  let fbv = \l.(fix \fib.\a.\n.
    new n1 in new n2 in new n3 in new n4 in
    n1 := n;
    if !n1 < 2 then 1
    else if !a(!n1) > 0 then !a(!n1)
    else (n2 := fib(a)(!n1 - 2);
          n3 := fib(a)(!n1 - 1);
          n4 := !n2 + !n3;
          a(!n1) := !n4;
          !n4))(mem)(l) in fbv(5)
\end{verbatim}

The examples above should serve to convince that \Verity\ is a conventional programming language with no hardware-specific primitives or constructs, although the type system has several subtle restrictions to ensure that the game-semantic models are finite-state, as discussed below.

\subsection{Theoretical and methodological background}

Compared to other higher-level academic or industrial synthesis tools the emphasis of \GoS\ is on correct and efficient support for the functional infrastructure of the language. Some restrictions are unavoidable because of the finite-state nature of the digital circuits, and the aim of \GoS\ is to impose no additional restrictions. It is a key methodological principle of the \GoS\ project that mature support for functions is essential in the pursuit of a useful and useable compiler. The theory behind \Verity-\GoS\ and the methodological considerations are explained at some length by~\cite{DBLP:conf/memocode/Ghica11}. 

At the most abstract level, digital circuits can be seen as topological diagrams of boxes and wires. The diagrams are topological (rather than geometrical) because in design we often wish to abstract from the size and length of the connectors, and from the precise placement of the components; such low-level matters are usually sorted out algorithmically by electronic design tools. One economical, elegant and mathematically canonical representation of diagrams is using combinators, which form a mathematical structure called a \emph{compact closed category}~\citep{KellyL80}. What is particularly useful about such a category in our context is that it can also describe a canonical model for a higher-order programming language with affine typing. This means that that the higher-order structure of the language is reflected directly in the diagrammatic structure of the circuit, which further means that abstraction and application can be represented with zero overhead.

From a practical point of view, the key consequence of the \GoS\ approach is that compiling a \Verity\ program produces a circuit with an interface determined by the type signature of the program. It is conventional to write the type of a program as a \emph{judgement} 
$x_1:T_1,\ldots,x_n:T_n\vdash P:T,
$
which says that program $P$ is well-typed of type $T$ and has free identifiers $x_i$ of type $T_i$. 

Each type corresponds to a circuit interface, defined as a
list of ports, each defined by data bit-width and a polarity. Every
port has a default one-bit control component. For example we write an
interface with $n$-bit input and $m$-bit output as $I=(+n, -m)$. More
complex interfaces can be defined from simpler ones using
concatenation $I_1\otimes I_2$ and polarity reversal
$I^\bullet=\mathsf{map}\,(\lambda x.{-}x)I$. If a port has only control and
no data we write it as $+0$ or $-0$, depending of polarity. Note that
obviously ${+}0\neq{-}0$ in this notation!

An interface for type $T$ is written as $\sbr T$, defined as
follows:
\begin{align*}
&  \sbr{\mathsf{com}} = (+0, -0) \qquad
   \sbr{\mathsf{exp}} = (+0, -n) \qquad
   \sbr{\mathsf{var}} = (+n, -0, +0, -n)\\
&  \sbr{T\times T} = \sbr{T}\otimes\sbr{T'} \qquad
  \sbr{T\rightarrow T'} = \sbr{T}^\bullet\otimes\sbr{T'}.
\end{align*}
The interface for commands \textsf{com} has two control ports, an input for starting execution and an output for reporting termination. The interface for integer expressions \textsf{exp} has an input control for starting evaluation and data output for reporting the value. Assignables \textsf{var} have data input for a write request and control output for acknowledgment, and control input for a read request along with data output for the value.  The tensor is a disjoint sum of the ports on the two interfaces while the arrow is like the tensor, but with a polarity-reversal of the ports occurring in the contra-variant position, as illustrated in the example below.

Diagrammatically, a list will correspond to ports from
left-to-right and from top-to-bottom. We indicate ports of zero
width (only the control bit) by a thin line and ports of width $n$ by
an additional thicker line (the data part). For example a circuit of
interface $\sbr{\mathsf{com}\rightarrow \mathsf{com}}=(-0,+0,+0,-0)$
can be written in any of these two ways in Fig.~\ref{fig:erp}.

\begin{wrapfigure}{l}{65mm}
    \includegraphics{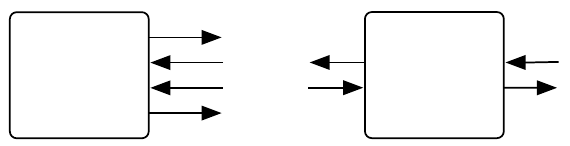}
    \caption{Equivalent representation of ports}
    \label{fig:erp}
\end{wrapfigure}

The unit-width ports are used to transmit \emph{events}, represented
as the value of the port being held high for one clock cycle. The
$n$-width ports correspond to data lines. We will work
under the assumption that the event on the unit port is a control
signal indicating the data on the data line is valid. 

The significant restriction that makes support for functions so simple is that the type system is \emph{affine}, which means that in function application the function and the argument cannot share free identifiers. This is an important restriction which has a major impact on the expressiveness of the language. For once, it is incompatible with imperative programming, in which variables naming memory locations must be reused in order to be read and written. 

In order to overcome this restriction we carefully add variable sharing to the programming language, using a type system called \emph{Syntactic control of concurrency} (SCC)~\citep{GhicaMO06}, which is based on Reynolds's \emph{Syntactic control of interference}~\citep{Reynolds78,OHearnPTT99}. The idea is to allow sharing of variables in product formation, but not in function application. Imperative sequential operations are then given uncurried type, so they can reuse variables. For example, the term \texttt{x:=!x+1} can be written, using a functionalised pre-fix notation as \texttt{assign (x, add(deref(x), 1))}. Assignment has type \texttt{assign : var * int -> com} and thus can share variables between its two arguments. An extra benefit of this type system is that by giving parallel command composition a curried type \texttt{par : com -> com -> com} it makes it impossible to have race conditions in the programming language, since the two arguments can never share identifiers. The program \texttt{c || c} (also written as \texttt{par c c}) does not type-check.

Unlike functions, variable sharing does not arise automatically out of the algebraic structure of the diagrammatic model. It needs to be implemented. Categorical considerations are nevertheless helpful in providing a family of equational specifications, corresponding to the notion of \emph{Cartesian product}, which establish that variable sharing is correctly implemented (because \emph{contraction} in the syntax corresponds to Cartesian product in the semantics). 

The conditions required for the correct implementation of product in \GoS\ amount to an input-output \emph{protocol} which all synthesised circuits must satisfy in order to compose properly. In the implementation, this protocol amounts to a simple \emph{bus protocol} needed for the correct time-multiplexed sharing of sequentially used circuits. They are formally described by~\cite{Ghica07}.

\begin{wrapfigure}[7]{r}[0pt]{70mm}
    \includegraphics{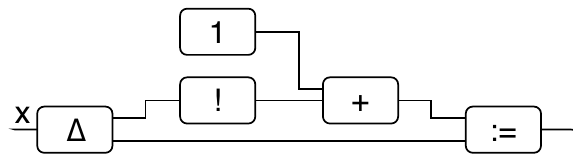}
    \caption{The diagonal circuit $\Delta$}
    \label{fig:diag}
\end{wrapfigure}
 
Diagrammatically, sharing is implemented by specialised circuits which correspond to the \emph{diagonal} in the Cartesian category of circuits. For example, the term $\mathtt{x:var\vdash x := {!}x+1:com}$ corresponds to the diagram sketched in Fig.~\ref{fig:diag}, with the diagonal labelled $\Delta$ used to share access to variable \texttt{x}.
 
To give an interactive, event based semantics to the imperative constants of \Verity\ we use the \emph{game semantics} of the language, which is formulated in this style. The interpretation of constants is standard in game semantics and will be not detailed here. To give a flavour of the implementation we show the iterator only (for convenience we have marked the top level ports, the ports of the loop guard and the ports of the body of the loop) in Fig.~\ref{fig:iter}, where OR joins two signals, T is a multiplexer and D is a unitary delay. This circuit can be realised either asynchronously~\citep{GhicaS10} or synchronously~\citep{GhicaM11}. Its input-output behaviour is:
\begin{itemize*}
\item receive an input signal from the top level;
\item propagate the input signal to the guard;
\item receive an input signal from the guard when it is ready;
\item use the data line from the guard in multiplexer $T$ to 
  \begin{itemize}
  \item propagate the signal back out to the top level if the guard is false;
  \item propagate the signal out to the loop body if the guard is true;
  \end{itemize}
\item use the termination acknowledgement from the body to trigger an evaluation of the guard again, which will cause the process to iterate. 
\end{itemize*}

\begin{wrapfigure}{r}{50mm}
    \includegraphics{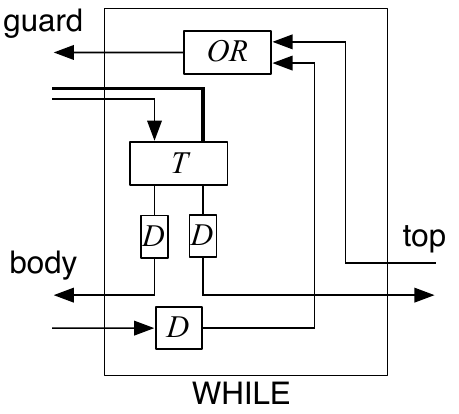}
    \caption{Iterator}
    \label{fig:iter}
\end{wrapfigure}
 
To have an expressive and convenient programming language the restrictions of the SCI system are still undesirable. They can be however avoided, to a great extent, using program analysis and transformation as described by~\cite{GhicaS11}. Finally, recursion can be implemented in the same framework, subject to several minor restrictions~\citep{GhicaSS11}. We do not describe these features in detail here as the complications they introduce are not directly relevant to tamper-proof compilation or coherent minimisation. The techniques described apply to these features as well. 

The compilation process is compositional and it allows the synthesis of circuits corresponding to \emph{open terms}. Compositionality in the compiler means that we have immediate support for \emph{separate compilation}. This is essential for having compiler support for (pre-compiled) libraries but, most importantly, for supporting \emph{foreign function interfaces} (FFI). Through the FFI we can interact with system-specific functionality which can be implemented outside of the programming language, using a conventional HDL. This is important as useful as low-level drivers for peripherals are written in HDL, but from the language we prefer to interact with them via function calls. 

Separate compilation and foreign function interface play a great role in making a compiler useful. However, interfacing with circuits produced outside the compiler exposes the synthesised code to low-level attacks, because such circuits cannot be assumed to satisfy the input-output protocol which synthesised circuits both satisfy and assume in order to operate properly.

\section{Protocols and low-level attacks} \label{sec:palla}

Tamper-proof compilation is relative to whatever notion of tampering we consider possible on pragmatic considerations, so a circuit is tamper resistant to the same extent as its physical substrate is. In other words, the high-level constraints needed for the proper operation of synthesised circuits cannot be violated without violating the underlying \emph{physical} constraints of the circuit. Note that some FPGA devices, such as Altera's \textsf{Cyclone III LS}, have physical anti-tamper layers which include special protection for the programming ports and redundancy checks.  \footnote{\url{http://www.altera.com/corporate/news_room/releases/2009/products/nr-ciii_ls.html}} 

\begin{wrapfigure}{l}{70mm}
    \includegraphics{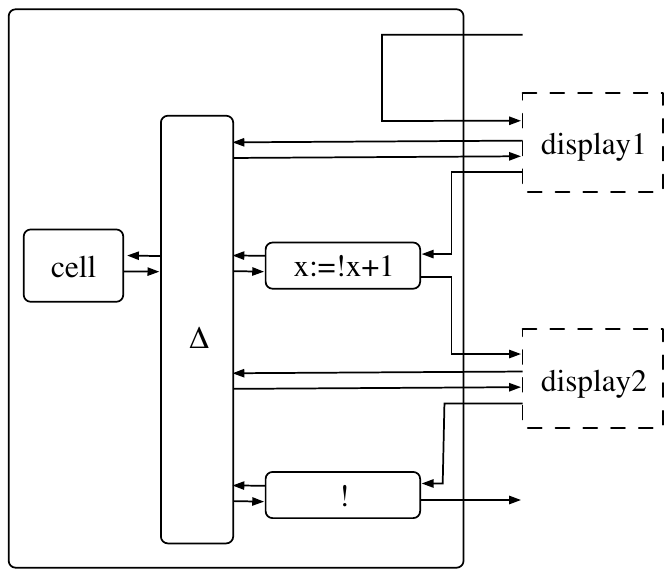}
    \caption{Example of synthesised program using the FFI}
    \label{fig:lla}
\end{wrapfigure}

Example of physical attacks on circuits involve over-heating or over-clocking the circuit so that it behaves erroneously on an electronic level. Also of a physical nature are observations against the temperature or energy consumption of the circuit as well as timing its responses. We provide no means of resistance against such attacks, but only against attackers which provide inputs and observe outputs at the ports only, within the normally accepted parameters of operation of the device. 

Let us illustrate the problem of low-level attacks with a very simple
example. Consider a program which interacts with the external
environment using functions \texttt{display1, display2:exp->com} to
drive, for example two segmented LED displays: 

\texttt{new x := 0 in display1(!x);}

\texttt{x:=!x+1; display2(!x); !x}

According to the semantics of \Verity\ this program should first
display the value 0 on device 1, then display value 1 on device 2,
then return the value 1 to the top level. The high-level diagram of
the concrete synthesised circuit in Fig.~\ref{fig:lla}.

The circuit labelled \texttt{x:=!x+1} is the incrementer in Fig.~\ref{fig:diag}. The circuit \texttt{cell} is a register storing the value of the variable and \texttt{!} is a de-referencer. The diagonal $\Delta$ shares access to \texttt{x} for the display functions, the incrementer and the final dereferencing. The dotted boxes are the implementations for the display functions, realised in HDL and visible from the programming language through the FFI. In order to simplify the drawing of the circuit the data lines are implicit where required; we only show the control lines.

Let us label the ports of the synthesised circuit top-to-bottom starting with 0 for the top-level port requesting execution and ending with 9 for the top-level port reporting the result. The correct interaction in which such a circuit is involved proceeds as follows:
\begin{enumerate*}
\item receive a top-level input request on port 0 to start execution;
\item request external function \texttt{display1} to execute using port 1;
\item using port 2 function \texttt{display1} may inquire what the value of its argument is, zero or more times;
\item using port 3 the circuit will always provide value 0 as response, the state of \texttt{cell};
\item eventually \texttt{display1} will terminate, reporting termination on port 4;
\item upon incrementing the register the circuit will use port 5 to request \texttt{display2} to execute;
\item  using port 6 function \texttt{display2} may inquire what the value of its argument is, zero or more times;
\item using port 7 the circuit will always provide value 1 as response, the new state of \texttt{cell};
\item eventually \texttt{display2} will terminate, reporting termination on port 8;
\item the circuit will report final value 1 on port 9. 
\end{enumerate*}
However, the environment, consisting of the top level and the two display functions can violate the input-output behavioural assumptions of synthesised code. Consider the environment in Fig.~\ref{fig:lla2}.
\begin{wrapfigure}[15]{r}{65mm}
  \includegraphics{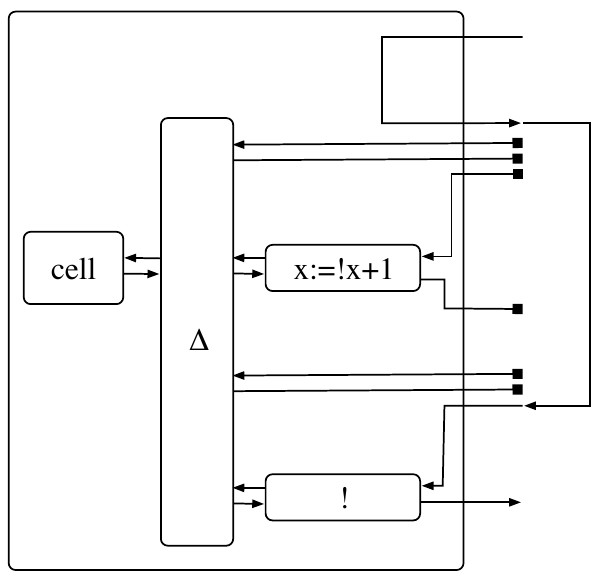}
  \caption{Environment which breaks language abstraction}
  \label{fig:lla2}
\end{wrapfigure}
The transaction in which the circuit is now involved is:
\begin{enumerate*}
\item receive an top-level input request on port 0 to start execution;
\item request external function \texttt{display1} to execute using port 1;
\item \texttt{display2} (illegally) reports termination on port 8;
\item the circuit will report final value 0, the initial state of the register, on port 9. 
\end{enumerate*}
The unused ports are marked as black squares for emphasis. 

The low-level attack violates the input-output behaviour of synthesised circuits, the essential features characterising correctly compiled \Verity\ programs, and it causes the program to produce the wrong value 0 instead of the expected value 1.  It is easy to see that the environment can manipulate the inputs and output to the two display functions so that the register \texttt{cell} and the final result can have  whatever value is desired by the attacker. Obviously, from a security point of view such tampering unacceptable as it can lead to a wide range of attacks against data integrity. 

\section{Enforcing programming language abstractions} 

Low level attacks are possible when the system can perform actions that break the programming language abstractions. But can we prevent the system from performing such actions? In this particular case the answer is positive. Programming language abstractions are reflected into the structure of the synthesised circuits in two ways: statically, as the input and output ports of the circuit, or dynamically, as the input-output behaviour of the environment in which the circuit operates.

\begin{wrapfigure}{l}{75mm}
  \includegraphics{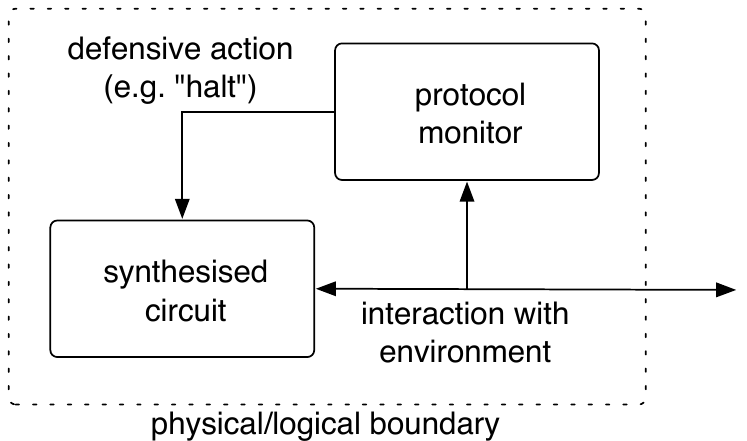}
  \caption{Architecture of a tamper-proof circuit}
  \label{fig:tpc}
\end{wrapfigure}

The static port structure cannot be violated, but the dynamic behaviour of the environment can violate the protocol-like semantics of the language. We can restrict the behaviour of the environment to legal traces by taking advantage of several facts~\citep{GhicaMO06}. First, we know what the \emph{fully abstract} model of \Verity\ is. A fully abstract model is a correct and complete characterisation of all the traces that can be generated by a synthesised \Verity\ program. Second, the fully abstract model of \Verity\ has a finite-state automaton representation for any type signature. Finally, the low-latency representation of the model, which is used for hardware synthesis, also has a finite-state representation~\citep{GhicaM11}.

The three observations above mean that all the legal interactions between a circuit and its environment can be described by a finite state machine, therefore by a digital circuit. In order to achieve tamper-proof-ness a synthesised circuit must not interact with its environment directly, but the interaction must be mediated by a monitor which will detect any illegal interactions and take appropriate actions if such illegal interactions occur. As illegal interactions indicate tampering attempts, the appropriate actions may be reset, halt, intentionally erratic behaviour or even destroying the circuit, depending on the level of protection and sensitivity desired. Schematically, the tamper proof circuit will look like in Fig.~\ref{fig:tpc}.

The precise specification of the legal interaction protocol and its low-latency asynchronous specification used for hardware synthesis are given by~\cite{GhicaM11} and will not be repeated here. For illustration we will detail the example of the previous section. The program has \emph{signature} \texttt{display1:exp$\rightarrow$com, display2:exp$\rightarrow$com $\vdash$ M:exp}, where $M$ is the program. To wit, the program uses two non-locally defined functions, \texttt{display1, display2} which are procedures taking integer expressions as arguments, and it has type expression. We write this signature as 
\[
\mathtt{(exp_1\rightarrow com_2)\rightarrow(exp_3\rightarrow com_4)\rightarrow exp_5}.
\]
The game-semantic model for \Verity\ stipulates that all legal traces in which the program can be involved have to have a form described by the Mealy machine in Fig.~\ref{fig:async}, which is dependent on the signature only:

\begin{wrapfigure}[13]{r}{70mm}
  \includegraphics[scale=0.75]{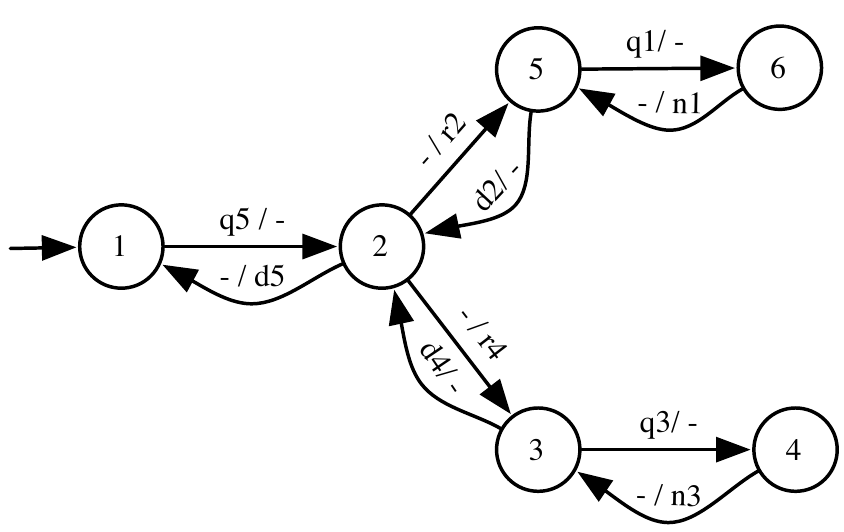}
  \caption{A game-semantic protocol, automaton representation}
  \label{fig:async}
\end{wrapfigure}

Intuitively, the reading of the protocol is this:
\begin{enumerate*}
\item The environment may start executing the program~(q5)
\item The program may terminate immediately (d5) or may ask for either of display functions to be evaluate (r2 or r4)
\item If a function was called by the program, it is allowed to either return immediately (d2 or d4) or it can evaluate its argument (q1 or q3) any number of times. 
\item The program must respond to a request to provide the argument of the function (n1 or n3).
\end{enumerate*}

The protocol above is \emph{asynchronous}, and for the purpose of hardware synthesis we use a low-latency representation called \emph{round-abstraction}~\citep{GhicaM11}, allowing multiple inputs and outputs on the same transition while avoiding deadlocks and race conditions, as in Fig.~\ref{fig:sync}.

\begin{wrapfigure}[15]{l}[0pt]{90mm}
  \includegraphics[scale=0.75]{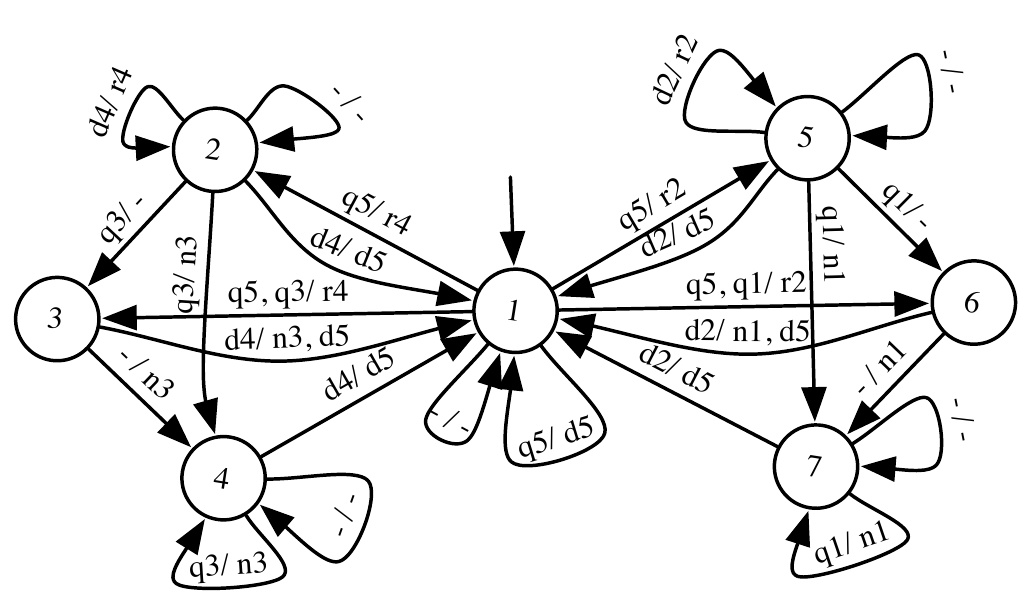}
  \caption{Synchronous representation of a protocol}
  \label{fig:sync}
\end{wrapfigure}

Producing a circuit representation of these finite-state machines is standard. Let us call this circuit \textsf{M} (monitor). The tamper-proof version of the circuit from the previous section is given in Fig.~\ref{fig:tpcex}.
The input lines to M, drawn in a lighter colour, are the interactions with the environment, and the output lines from M trigger the reset lines of the component sub-circuits. Note that the monitor M can also be placed in series (rather than in parallel) with the monitored circuit so that tampering signals never reach it. This is only marginally safer but comes at a cost of increased latency.  

With the tamper-proof version of the circuit the attacks of the previous section trigger resets (or whatever other defensive anti-tampering behaviour is desired in the concrete implementation) and are rendered ineffective. 

\section{Coherent equivalence and minimisation}  %2p

By constraining the interaction between the circuit and the environment, the tamper-proof compiler makes possible more aggressive optimisations. In conventional automata optimisation two states are considered equivalent if they are not distinguishable by any environment; this concept is formalised by \emph{bisimulation}. Bisimilar states can be identified, leading to optimised automata with fewer number of states. 

In the case of the tamper-proof compiler, the interaction between the circuit (an automaton) and the environment is monitored so that only certain interactions are permitted. This makes the environment less discriminating, leading to more states being equivalent. In the limit case of an empty protocol (no interactions are allowed) for example, all states of an automaton are equivalent and can be identified. Conversely, in the other limiting example of the protocol that permits all interactions, this notion of equivalence reduces to conventional bisimulation. 

\begin{wrapfigure}{r}{60mm}
  \includegraphics[scale=0.75]{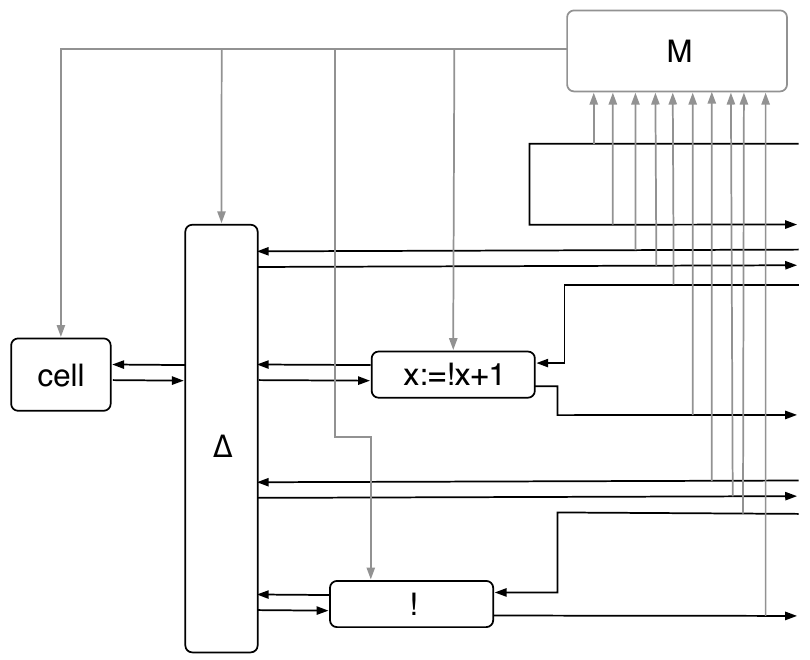}
  \caption{Tamper-proof compiled circuit with monitor}
  \label{fig:tpcex}
\end{wrapfigure}

In this section we formalise the concept of \emph{coherent equivalence} for finite-state transducers and we formulate (the obvious) minimisation algorithm based on this notion of equivalence, proving its correctness. In the following section we discuss the way it is incorporated in the \Verity/\GoS\ compiler. A similar notion of coherent bisimulation could be formulated but it is more complex; since we are only interested in finite traces we will not pursue this theoretical development here. 

Let a \emph{transducer signature} $A$ be a pair of finite sets of labels $(I_A, O_A)$, the input and the output ports of a transducer, respectively. By $A^\bullet$ we mean signature $A$ with input-output port polarities reversed. We call a (possibly empty) subset of $A$ a \emph{round}, and an occurrence of a label in a round an \emph{event}. Let a \emph{synchronous trace} over signature $A$ be a sequence of rounds. Let $\epsilon$ be the \emph{empty trace (empty sequence)}.  
\begin{definition}[Transducer]
  A \emph{transducer} with signature $A$, written $T:A$ is a triple $\langle S,s,\delta\rangle$ where $S$ is a set of states, $s^0\in S$ is the initial state and $\delta\subseteq S\times\mathcal P(A)\times S$ is a transition relation. 
\end{definition}
We often write $S_T$ to mean ``the set of states $S$ of transducer~$T$'', etc. We extend the transition relation $\delta$ to traces $\hat\delta\subseteq S\times\mathcal P(A)^*\times S$ in the usual way:
\begin{itemize*}
\item for any state $q\in S$, $(q,\epsilon,q)\subseteq\hat\delta $
\item For any round $V\subseteq A$, states $s,s'\in S$ and trace $t\in\mathcal P(A)^*$, $(s, t\cdot V, s')\in\hat\delta$ if and only if there is a state $s''\in S$ such that $(s,t,s'')\in\hat\delta$ and $(s'',V,s')\in\delta$. 
\end{itemize*}
Let $\sbr{T:A}$ be the set of traces of a transducer, 
$
\sbr{T:A}\stackrel{\text{def}}=\{
t\in\mathcal P(A)^* \mid \exists s\in S.(s^0, t, s)\in\hat\delta
\}
$.
Given a state $s\in S$ for a transducer, it is useful to know the set of witness traces which can reach this state in the transducer,
$
\omega_T(s)\stackrel{\text{def}}=\{
t\in\sbr T\mid 
(s^0,t,s)\in\hat\delta
\}.
$

We say that two transducers with the same signature $T,T':A$ are \emph{equivalent} if and only if they have the same set of traces, $T\equiv_A T'$ iff $\sbr T=\sbr{T'}$. This is the conventional notion of equivalence, and it is preserved by all common operations on transducers, such as intersection, product, etc. 
\begin{definition}[Intersection]\label{defn::t-intersection}
  Given transducers $T,T':A$, we define 
  $
  T\cap T'\stackrel{\text{def}}=\langle
  S_T\times S_{T'},
  (s_T^0,s_{T'}^0),
  \delta
  \rangle
  $, with
  $
  ((s_1,s_1'),V,(s_2,s_2'))\in\delta \text{ iff } 
  (s_1,V,s_2)\in\delta_T \text{ and }
  (s_1',V,s_2')\in\delta_{T'}.
  $
\end{definition}
This definition is sound in the sense that for any transducers $T,T':A$ we have that $ \sbr{T \cap T'}=\sbr{T} \cap \sbr{T'}$. The proof is immediate. 
We define a laxer notion of equivalence motivated by a restricted set of interactions between the transducer and its environment. Let us define this restricted set of interactions by $P:A$, also a transducer, which we call a \emph{protocol}. 
\begin{definition}[Coherent equivalence]\label{defn::t-coherent-equivalence}
  We say that transducers $T,T':A$ are \emph{coherently equivalent} under protocol $P:A$, written $T\equiv_A^P T'$ if and only if $T\cap P\equiv_A T'\cap P$. 
\end{definition}
The main definition in this section identifies when two states are coherent, i.e. equivalent under a restricted set of observations.

\begin{definition}[Coherent state simulation]\label{defn::coherent-refine}
Given a transducer $T:A$, a protocol $P:A$ and a relation $R \subseteq S_T \times S_T$, we say that $R$ is a \emph{coherent simulation}, iff, for any $(s',s'') \in R$ the following two conditions hold:
\begin{enumerate}
\item For any transition $(s'', V, r'') \in \delta_T$, there exists a transition $(s', V, r') \in \delta_T$ such that $(r', r'') \in R$;
\item For any round $V \subseteq A$, such that there is no transition from $s''$ labelled with $V$,
\begin{enumerate}
\item either there is no transition from $s'$ labelled with $V$,
\item or $\omega_T (s'')\cdot V \cap \sbr P = \emptyset.$
\end{enumerate}
\end{enumerate}
\end{definition}
If $(s',s'')\in R$ for some protocol $P$ we write $s'\smallfrown^P s''$.

%\begin{definition}[Coherent state refinement]\label{defn::coherent-refine}
%  Given transducer $T:A$, protocol $P:A$ and states $s,s'\in S_T$ we say that state $s$ \emph{coherently refines} $s'$, written $s\smallfrown^P_A s'$ if and only if for all rounds $V\subseteq A$ such that there is state $r\in S_T$ such that $(s,V,r)\in\delta_T$:
%  \begin{enumerate*}
%  \item if there is state $r'\in S_T$ such that $(s', V, r')\in\delta_T$ then $r\smallfrown^P r'$, or
%  \item if $\forall r'\in S_T, (s',V,r')\not\in \delta_T$ and $\forall U \subseteq O_A$, $\forall r''\in S_T, (s', I_V\cup U),r'') \not\in\delta_T$ then $\omega_T(s')\cdot V\cap\sbr P=\emptyset$.
%  \end{enumerate*}
%\end{definition}

\begin{definition}[Coherent state equivalence]\label{defn::symmetric-property}
  Given transducer $T:A$, protocol $P:A$ and states $s',s''\in S_T$ we say they are \emph{coherently equivalent}, written $s'\asymp_A^P s''$ if and only if $s'\smallfrown^P_A s''$ and $s''\smallfrown^P_A s'$.
\end{definition}
Given a function $f:A\rightarrow B$ we denote by $(f\mid x\mapsto y):A\cup\{x\}\rightarrow B\cup\{y\}$ the function which maps $x$ to $y$ and otherwise behaves like $f$. We denote by $id_A:A\rightarrow A$ the identity function on $A$, omitting the subscript if clear from context.

We call the transducer obtained by identifying two states a \emph{quotiented} transducer, defined as follows.
\begin{definition}[Quotienting]\label{defn::t-quotienting]}
  Given transducer $T=\langle S\uplus\{s_1,s_2\}, s^0,\delta\rangle$ we define its \emph{quotient} $T/(s_1,s_2)$ as follows:
  \begin{itemize*}
  \item $S_{T/(s_1,s_2)}=S\uplus\{s\}$
  \item $s^0_{T/(s_1,s_2)}= (id_S\mid s_1\mapsto s\mid s_2\mapsto s)(s^0)$
  \item $(r_1, V, r_2)\in \delta_{T/(s_1,s_2)}$ iff there are $r_i'\in(id_S\mid s_1\mapsto s\mid s_2\mapsto s)^{-1}(r_i)$, $i=1,2$ such that $(r_1', V, r_2')\in \delta$.
  \end{itemize*}	
\end{definition}

\begin{lemma}\label{lemma::T subset from the quetiented}
For any transducers $T,T/(s',s''):A$ we have $\sbr{T} \subseteq \sbr{T/(s',s'')}$.
\end{lemma} 

The proof is immediate, as the quotiented transducer always introduces new transitions while preserving the old ones. 

If two coherently equivalent states, under some protocol, are quotiented in a transducer then the resulting transducers are coherently equivalent under the same protocol. 
\begin{theorem}[Soundness]
For any transducer $T:A$, protocol $P:A$, and states $s',s''\in S_T$, if $s'\asymp_A^T s''$ then $T \equiv_A^P T/(s',s'')$. 
\end{theorem}

\begin{proof}
By expanding Def.~\ref{defn::t-coherent-equivalence}, we need to show that $T\cap P\equiv_A T/(s',s'')\cap P$. Lem.~\ref{lemma::T subset from the quetiented} proves one direction. The other direction, by soundness of $\cap$, is $(\sbr{T/(s',s'')}\setminus\sbr{T})\cap \sbr{P}= \emptyset$. 
We prove it by induction on the length of trace $t\in\mathcal P^*(A)$. Let $t \in\sbr{T/(s',s'')}$. We need to show that either $t \in \sbr{T}$, or $t \notin \sbr{T}$ and $t \notin \sbr{P}$. From the definition of $\sbr{T/(s',s'')}$, $\exists q\in S_{T/(s',s'')}$ such that $(s_{T/(s',s'')}^0, t, q)\in\hat{\delta}_{T/(s',s'')}$. According to Def.~\ref{defn::t-quotienting]} if $s_T ^0 = s_{T/(s',s'')}^0$ then this implies that $t \in \sbr{T}$, otherwise $s_T^0 =s'$ and $s_{T/(s',s'')}^0 = s$ or $s_T^0 =s''$ and $s_{T/(s',s'')}^0 = s$. Clearly in the latter two cases we have $t \notin \sbr{T}$, but since we assumed $s\asymp_A^T s'$ then by  Def.~\ref{defn::coherent-refine} and Def.~\ref{defn::symmetric-property} and by considering that $\epsilon$ is a witness trace for the start state $s_T^0 $, for both cases, we get $t \notin \sbr{T}$ and $t \notin \sbr{P}$. So the property holds for every trace of length one. 

For the inductive step we want to show that for any trace $t\cdot V$ if $t\cdot V \in \sbr{T/(s',s'')}$ then $t\cdot V \in \sbr{T}$ or $t \cdot V \notin \sbr{T}$, and $t\cdot V\notin \sbr{P} $. First note that if $t \cdot V \in \sbr{T/(s',s'')}$ then $t \in \sbr{T/(s',s'')}$, since transducer languages are prefix-closed. Applying the induction hypothesis, either $t \in \sbr{T}$ or $t \notin \sbr{T}$, and $t \notin \sbr{P}$. For the first case, as in the base case, unfold Def.~\ref{defn::t-quotienting]} and take $t$ as a witness trace;  we deduce that if the target state of the trace $t$ is $s'$, respectively $s''$, and the source state of of round $V$ is $s''$, respectively $s'$, then this implies that $t \cdot V \notin \sbr{T}$, but since we we assumed earlier in this proof that $s\asymp_A^T s'$ and by unfolding  Def.~\ref{defn::coherent-refine} and Def.~\ref{defn::symmetric-property} then this means $ t \cdot V \notin \sbr{P}$. For the second case the proof is immediate from prefix closure, since either $t\cdot V \in \sbr{T}$ or $t \cdot V \notin \sbr{T}$, and $t\cdot V\notin \sbr{P}$. 
\end{proof} 

Note that if the protocol $P:A$ is the trivial protocol which accepts all interactions then coherent equivalence and quotienting become the conventional notions of equivalence and minimisation. 

The soundness theorem states that an environment constrained by the protocol cannot distinguish between the original and the quotiented transducer, which has a smaller number of states. Iteratively quotienting all pairs of coherently equivalent states produces a \emph{coherently minimised} transducer. 

Because we are working in a compiler, the issue of compositionality is also important. The interaction protocols between the program and the environment are dictated by the type signature of the program, therefore different programs will observe different protocols. The following result shows that coherent minimisation is not only sound, but also \emph{compositional}, i.e.\ it can be applied to any sub-component of a larger system without affecting its overall properties, including coherence equivalence itself.

We denote by $t\upharpoonright A$ a trace generated by deleting all port labels that do not belong to signature $A$. The definition of projection ($\upharpoonright$) can be lifted to sets by pointwise application. We denote by $\theta\subseteq \mathcal P(A)^*$ sets of traces over signature $A$. 
We define the \emph{interaction} of two set of traces $\theta\subseteq\mathcal P(A+B)^* $ and $\theta'\subseteq\mathcal P(B+C)^*$ as $\theta \parallel \theta' =\{ t \in \mathcal{P}(A+B+C)^* \mid t\upharpoonright A+B \in \theta\wedge t\upharpoonright B+C \in \theta'\}$. \emph{Composition} of sets of traces is defined as interaction followed by hiding (projection), which is the standard definition in trace-based models of processes $\theta;\theta'=\{t\in \mathcal{P}(A+C)^* \mid t \in \theta \parallel \theta' \}$.

\begin{definition}[Transducer interaction]\label{defn::t-interaction}
  Given transducers $T:A+B$ and $ T: B+C$, we define 
  $
  T\parallel T':A+B+C=\langle
  S_T\times S_{T'},
  (s_T^0,s_{T'}^0),
  \delta
  \rangle
  $, where 
  $
  ((s_1,s_1'),V,(s_2,s_2'))\in\delta$ iff  
  $(s_1,V\upharpoonright A+B,s_2)\in\delta_T$ and 
  $(s_1',V\upharpoonright B+C,s_2')\in\delta_{T'}.
  $
\end{definition}
\begin{definition}[Transducer projection] 
Given $T:A+B$ we define $T\upharpoonright A=\langle S_T,s_T^0,\delta\rangle$, where 
$  (s_1,V\upharpoonright A,s_2)\in\delta$ iff 
$  (s_1,V,s_2)\in\delta_{T}$.
\end{definition}
\begin{definition}[Transducer composition]
Given transducers $T:A+B$ and $ T: B+C$, we define 
  $
  T; T'=(T\parallel T')\upharpoonright A+C$.
\end{definition}
\begin{theorem}[Soundness of composition]\label{t::comp-soundness}
For any two transducers $T:A+B$ and $T':B+C$ we have 
$\sbr{T||T'}=\sbr{T}||\sbr{T'}$
and
$\sbr{T;T'}=\sbr{T};\sbr{T'}$.
\end{theorem}
\begin{theorem}[Compositionality]
For any transducers $T,T':A+ B, T'':B+ C$ and protocols $P:A+ B, P':B+ C$ if $T\equiv_{A+ B}^P T'$ then $T;T''\equiv_{A+ C}^{P;P'} T';T''$.
\end{theorem}
\begin{proof}
 Let $T\equiv_{A+ B}^P T'$ which by Def.~\ref{defn::t-coherent-equivalence} is equivalent to $T \cap P = T' \cap P $. By conventional equivalence and the soundness of transducer intersection, $\sbr{T}\cap \sbr P = \sbr{T'} \cap \sbr{P}$. We want to show that $T;T''\equiv_{A+ C}^{P;P'} T';T''$, which by Def.~\ref{defn::t-coherent-equivalence} is  $\sbr{(T;T'')\cap (P;P')} =\sbr{(T';T'')\cap (P;P')}$, which we do by double inclusion.
$\sbr{(T';T'')\cap (P;P')} \subseteq \sbr{(T;T'')\cap (P;P')}$. Let $t \in \sbr{(T';T'')\cap (P;P')}$, which by the soundness of transducer intersection and composition is equivalent to $t \in \sbr{T'};\sbr{T''}$ and $t\in\sbr{P};\sbr{P'}$, so there exists a trace $t':A+B+C$ such that  $t'\upharpoonright A+C = t$ and $t' \in \sbr{T'}\parallel \sbr{T''}$ and $ t' \in \sbr{P}\parallel\sbr{P'}$. By the definition of trace-set interaction it follows that $t'\upharpoonright A+B \in \sbr{T}$ and $t'\upharpoonright A+B\in\sbr{P}$ and $ t'\upharpoonright B+C \in \sbr{T''}$ and $t'\upharpoonright B+C\in \sbr{P'}$. Since earlier we have shown  $\sbr{T}\cap \sbr{P} = \sbr{T'} \cap \sbr{T'}$ it immediately implies that $t'\upharpoonright A+B \in \sbr{T}$, which means that $t' \in \sbr{T}\parallel \sbr{T''}$. By following the definition of trace-set composition we deduce that $t'\upharpoonright A+C \in \sbr{T};\sbr{T''}$. Since $t = t'\upharpoonright A+C$, $t \in \sbr{T};\sbr{T''}$, and hence we have already $t \in \sbr{P};\sbr{P'}$ we conclude that $t \in (\sbr{T};\sbr{T''}) \cap (\sbr{P};\sbr{P'})$ which by the soundness of transducers intersection and composition is  equivalent to $t \in \sbr{(T;T'')\cap (P;P')} $. The other direction is similar.
\end{proof}

\section{Symbolic transducers and synthesis}
Finite-state transducers are unsuitable in dealing with numbers due to the very large numbers of states required. Transducers interpret the entire state space explicitly, and hence, it is too computationally expensive to construct them for arbitrary values. A standard technique is to use symbolic representations to overcome these limitations. Several transitions from one source state to different target states can be combined into a single transition governed by a symbolic condition, a \emph{guard}. A symbolic finite-state transducer (SFST) uses two components to represent state: a finite set of \emph{control states} and a finite set of \emph{registers} of unbounded size to handle data. Registers have initial values and can be modified explicitly via symbolic expressions (\emph{updates}). In our concrete implementation, symbolic guards and updates use the expressions syntax of YICES, a state of the art of tool for solving satisfiabilty modulo theory (SMT).~\footnote{\url{http://yices.csl.sri.com/}} Generating synthesisable HDL descriptions of circuits from symbolic transducers is straightforward. 

\begin{definition}[SFST]
 A Symbolic Finite State Transducer $T$ over port signature $A$ written $T:A$ is defined by as a tuple $T\stackrel{\text{def}}= \left\langle S_T,Reg_T, s^0, \delta_T \right\rangle $, where:
\begin{itemize*}
\item  a finite set of \emph{states}, $S_T$
\item  a finite set of \emph{local stores} (\emph{registers}) $Reg_T$
\item  a designated \emph{start state}, $s^0\in S_T$ 
\item  A \emph{transition relation}, $\delta_T \subseteq S_T \times \mathcal{P}(A) \times \Gamma \times \Delta \times S_T$ 
\end{itemize*}
\end{definition}
Note that the \emph{guard} and the \emph{update} are symbolic expressions specified in a fixed grammar, in our case that of YICES. Let $\Gamma$ and $\Gamma'$ be the set of guard and update expressions, respectively and let $\Delta \subseteq Reg\times\Gamma'$ be a language of \emph{updates} which includes the identity. Every register $x \in Reg$ is assumed to have initial value 0.

As a simple example, a SFST which reads an integer twice from port $x$ and outputs the sum of the two values, only if it is positive, to port $r$ can be defined as below.
\[
T=\langle
\{A, B, C\},
\{y, z\},
A,
\{
(A, \{x\}, true, y\leftarrow x, B),
(B, \{x\}, true, z\leftarrow x, C),
(C, \{r\}, y+z>0, r\leftarrow y+z, A)
\}
\]
All definitions from the previous section lift in the obvious way to SFSTs since their correctness is not contingent on the set of states being finite, and since all SFSTs can be obviously represent explicitly by mapping the register values into a concrete state. The key difference between FSTs and SFSTs is in their computational properties, which for the latter can be undecidable.  For this reason, the key notion of coherent equivalence and quotienting only apply to control states, rather than symbolic states encompassing both control states and register values. For computational reasons we also restrict the notion of protocol to the order in which ports are activated, ignoring the values on the ports. So, as a protocol we could specify that an operation reads the input twice then produces output, but we cannot specify that it is an adder. 

\begin{definition}[Symbolic protocol]
  A SFST $T$ is a \emph{symbolic protocol} only if $\Gamma=\{\mathit{true}\}$ and $\Delta=id$. 
\end{definition}
\begin{definition}[SFST-Coherent state simulation]\label{defn::sfst-coherent-refine}
  Given a SFST $\text{ } T:A$, a symbolic protocol $P:A$ and a relation $R \subseteq S_T \times S_T$, we say that $R$ is a \emph{coherent simulation}, iff, for any $(s',s'') \in R$ the following two conditions hold:
  \begin{enumerate}
  \item For any transition $(s'', V, g'', \Delta'', r'') \in \delta_T$, there exists a transition $(s', V, g', \Delta', r') \in \delta_T$ such that $(r', r'') \in R \text{ and } g'\equiv g'' \text{ and } \Delta' \equiv \Delta''$;
  \item For any round $V \subseteq A$, such that there is no transition from $s''$ labelled with $V$,
\begin{enumerate}
\item either there is no transition from $s'$ labelled with $V$,
\item or $\omega_T (s'')\cdot V \cap \sbr P = \emptyset.$
\end{enumerate}
\end{enumerate}
\end{definition}

\begin{figure}
  \begin{center}
    \includegraphics[scale=.8]{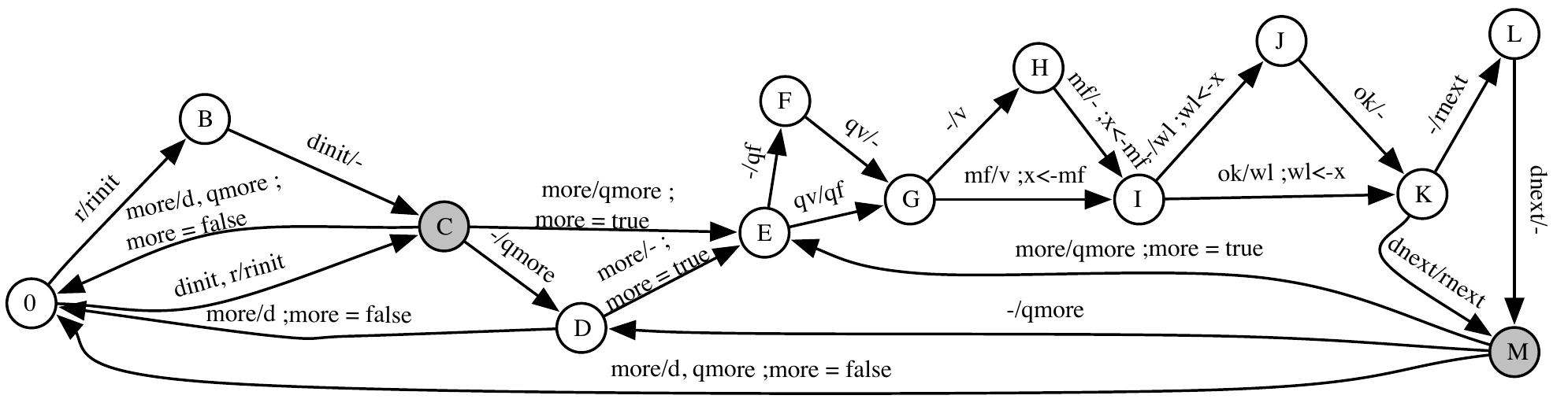}
  \end{center}
  \caption{Original representation of in-place map}
  \label{ref:orig}
\end{figure}

\subsection{Example: coherent optimisation of in-place map}
\begin{wrapfigure}[20]{l}[0pt]{60mm}
  \includegraphics[scale=.8]{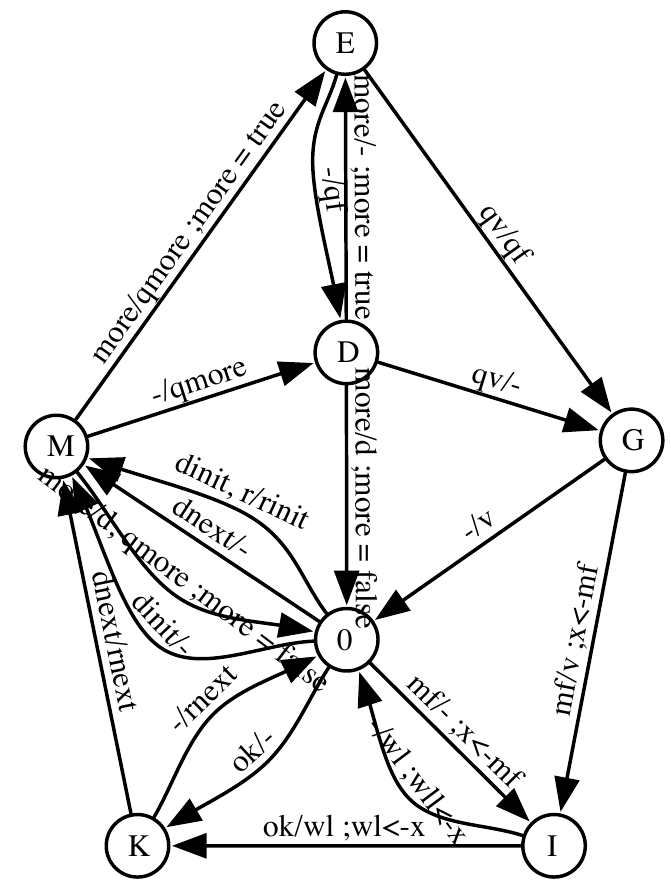}
  \caption{Coherently minimised transducer}
  \label{fig:minim}
\end{wrapfigure}
 
Consider a \Verity\ program which applies a function \texttt{f} to a data structure through an iterator: \texttt{init} initialises it, \texttt{more} tests whether there are more elements, \texttt{l} gives write-access to the current location, \texttt{v} gives the value of the current location and \texttt{next} moves to the next location. The code~is:
\\
\texttt{
\mbox{}\hspace{3ex} import f, init, more, l, v, next;\\
\mbox{}\hspace{3ex} init; while more do \{l := f(v); next\}\\
}

The (asynchronous) protocol $P$ that is considered here for the minimisation process, given by the type signature of the program, can be described by the following regular expression: $(r ( q_{more}  b_{more} + q_f^1 (q_f^2 m_f^2)^* m_f^1 + r_{init}  d_{init} + r_{next} d_{next} + w_{l}^n  ok_{l} + q_{v} m_{v})^* d)^*$. The synchronous version is derived from this using round abstraction and is significantly more complicated. 

The compilation of the program gives the SFST in Fig.~\ref{ref:orig}, with 13 states. Because the protocol is enforced by the monitor we can apply coherent minimisation, which results in a version with 7 states (Fig.~\ref{fig:minim}). The tuples of coherently equivalent states are $D\asymp^PF$, $0\asymp^PB\asymp^PH\asymp^PJ\asymp^PL$, $C\asymp^PM$, and can be quotiented away.  Using conventional bisimulation, such as implemented by the Hopcroft minimisation algorithm, we get an automaton with 12 states as only $C$ and $M $ are bisimilar. 
\\[2ex]

\section{Towards practical and efficient secure computation} 

Higher-level synthesis via \GoS\, augmented with the protocol-monitoring mechanism of Sec.~\ref{sec:palla}, can produce tamper-proof instances of arbitrary libraries with rich functional interfaces. This can offer a new approach to the problem of collaborative computation without data disclosure, i.e.\ \emph{secure computation}~\citep{Ben-Or:1988:CTN:62212.62213}. Much research in this area is concerned with writing programs that do not inadvertently leak information, e.g.\ privacy-preserving data mining~\citep{Agrawal00}, but system-level security guarantees in the form of absence of low-level attacks and exploits are essential to make this practical. Because system-level security is difficult to guarantee on the desktop, secure computation is generally thought of in the context of distributed and cloud computing, where the physical separation of resources makes system-level security guarantees easier to achieve.

On the desktop, on the other hand, secure computation requires the presence of a trusted hardware module that can prevent sensitive data (keys, value registers, etc.) from being tampered with, the typical example being the \emph{Trusted Platform Module} (TPM)\footnote{See \url{http://www.trustedcomputinggroup.org/}}. A TPM can also be used to authenticate arbitrary binary code and authorise its access to sensitive data, so it is possible to set up a secure computation framework using it. The most practical way of doing this is using virtualisation to set up a secure virtual machine for the execution of trusted code, and interfacing it with the rest of the machine as if it was a distinct physical computer. This works because modern processors support virtualisation natively and protect the memory space of the virtual machine. 

Hardware compilation is a way to produce fully customised secure hardware modules that can interface with the rest of the system using a convenient higher-order interface. Low-level attacks and exploits on the module are prevented first via the physical tamper-proof mechanism of the FPGA fabric, and logically through a monitoring mechanism which prevents interactions that do not respect the programming language protocol. This allows properties established by reasoning at the programming languages level to be guaranteed in the implementation. Compared with TPM-based approaches, this approach has two potential advantages. First, is its simplicity. No special TPM is required and no native virtualisation support is needed in the untrusted device. Through higher-level synthesis we can produce special-purpose devices that provide only restricted functionality. Verifying the logical security properties of such devices is significantly easier than verifying the security properties of general-purpose software and hardware mechanisms such as TPMs and virtualisation frameworks. The second advantage is that of low overhead. TPM-based secure computation on the desktop involves a significant amount of overhead, as does the communication between the secure virtual machine and the rest of the system. On the other hand a FPGA can be set up to interface with a CPU on a physical level via the system bus; a variety of FPGA-based PCI cards are commercially available and can be used to implement this system. This is further work. 

We also note that once the interaction between a system and its environment can be effectively constrained, more states in the system are observably equivalent. This novel notion of equivalence, called coherent equivalence, can be used to aggressively optimise such systems by reducing the number of states. Coherent equivalence is, of course, not restricted to our particular semantic model, which is meant to serve as motivation and illustration, but to any automata operating in restricted environments. For example, APIs themselves can be enforced by a monitor, stipulating that functions belonging to the API must be called in a particular order. 

Our formulation is a first step, but more work needs to be done: the main definition (Def.~\ref{defn::coherent-refine}) has a formulation which relies on the trace-semantic interpretation of the transducer; a direct definition similar to that of bisimulation would be preferable. This is also future work. Most importantly, the simulation relation of Def.~\ref{defn::coherent-refine} is in general not transitive. This raises one technical problem which is somewhat unpleasant although not critical. Our state-reduction algorithm will lead to results which are dependent of the order in which coherently equivalent states are identified and quotiented.  If $s_1\frown s_2$ and $s_2\frown s_3$ and if we quotient $s_1,s_2$ (as a new state $s'$) we don't know a priori whether $s'\frown s_3$. From experiments we know that sometimes it holds sometimes it doesn't. Practically, this makes the state-reduction algorithm slower because $\frown$ needs to be recalculated after each quotienting.

%\appendix
%\section{Appendix Title}

%This is the text of the appendix, if you need one.

%\acks

%Acknowledgments, if needed.

% We recommend abbrvnat bibliography style.

\bibliographystyle{abbrvnat}

\bibliography{ghica}

% The bibliography should be embedded for final submission.

\end{document}